\documentclass{article}

\usepackage{arxiv}

\usepackage[utf8]{inputenc}
\usepackage[T1]{fontenc}
\usepackage{hyperref}
\usepackage{url}
\usepackage{booktabs}
\usepackage{amsfonts}
\usepackage{amsmath}
\usepackage{amssymb}
\usepackage{amsthm}
\usepackage{nicefrac}
\usepackage{microtype}
\usepackage{xcolor}
\usepackage{graphicx}
\usepackage{algorithm}
\usepackage{algorithmic}
\usepackage{natbib}
\usepackage{graphicx}
\usepackage{multirow}
\usepackage{makecell}
\usepackage{threeparttable}
\usepackage{ulem}
\usepackage{float}
\usepackage{enumerate}

\newtheorem{theorem}{Theorem}

\title{Trustworthy AI/ML Regression and Unbiased Causal Inference 
for Real-World  Data }

\author{
  \textbf{Yifei Xu}$^{1,2,3}$, \textbf{Hwiyoung Lee}$^{1,2,3}$, \textbf{Zhenyao Ye}$^{1,2,3}$, \textbf{Yezhi Pan}$^{4}$, 
  \textbf{Jingsong Zhou}$^{4}$, \textbf{Yun Yang}$^{4}$, \\[4pt] 
  \textbf{Chixiang Chen}$^{1,3}$, \textbf{Shuo Chen}$^{1,2,3,*}$ \\[4pt]
  $^{1}$Division of Biostatistics and Bioinformatics, Department of Epidemiology and Public Health,\\
  School of Medicine, University of Maryland\\[2pt]
  $^{2}$Maryland Psychiatric Research Center, Department of Psychiatry,\\
  School of Medicine, University of Maryland\\[2pt]
  $^{3}$The University of Maryland Institute for Health Computing (UM-IHC)\\[4pt]
  $^{4}$Department of Mathematics, University of Maryland, College Park\\[4pt]
  \texttt{$^{*}$shuochen@som.umaryland.edu}
}

\begin{document}

\maketitle

\begin{abstract}
Real-World Data (RWD), with its large sample sizes and rich clinical detail, offers a compelling alternative to randomized controlled trials (RCTs) for studying treatment effects in diverse and complex patient populations. However, its observational nature introduces confounding that prevents straightforward comparative effectiveness research. Target trial emulation leverages RWD to estimate average treatment effects (ATE) at the population scale and diversity that RCTs cannot achieve, yet its validity depends critically on unbiased ATE estimation under high-dimensional confounding. Many causal inference pipelines address high-dimensional confounding through machine learning and artificial intelligence (ML/AI) outcome regression. However,  commonly used ML/AI regression models exhibit systematic prediction bias, with predicted outcomes shrinking toward the marginal outcome mean. This structural bias propagates into ATE estimation and cannot be corrected by cross-fitting, ensemble methods, or any standard ML practice. In this work, we first quantitatively characterize how systematic prediction bias in ML/AI outcome regression leads to biased ATE estimates in causal inference models. We further propose an unbiased ML/AI regression-based causal inference framework to ensure unbiased ATE estimation for observational studies. We demonstrate our approach by studying the effects of opioids on cardiovascular health in patients with chronic pain using UK Biobank data.
\end{abstract}
% key words. average treatment effect,  continuous outcome, high-dimensional covariates,    systematic prediction bias, target trial emulation
%-----------------------------------------------------------------------
\section{Introduction}
\label{sec:intro}
%-----------------------------------------------------------------------
% \textcolor{red}{good to include more references, e.g. to RWE, RCT, and any claims made here.}
Real-World Data (RWD) is routinely generated from a variety of healthcare sources, including electronic health records (EHRs), administrative claims databases, disease registries, and wearable device monitoring, capturing patient experiences outside the controlled setting of clinical trials. Such data are then used to generate Real-World Evidence (RWE), which helps evaluate the safety and effectiveness of medical products across diverse patient populations \citep{hubbard2024targettrial, antoine2023breastcancer, fu2023targettrial, kuehne2022ovarian}. While randomized clinical trials (RCTs) remain the gold standard for
establishing causal treatment effects, their practical limitations are well recognized: they are
costly, time-consuming, and their strict eligibility criteria may systematically exclude 
patients with multiple comorbidities and underrepresented minority populations, precisely the
groups bearing the greatest disease burden in routine practice.

Target trial emulation (TTE) provides a principled framework for bridging this gap
\citep{hernan2016using}. By explicitly mapping the design features of a hypothetical RCT onto
an observational dataset, TTE guards against common biases in observational research and enables
studies of treatment effects at a scale, duration, and population breadth that no RCT could
feasibly achieve.  {Unbiased estimation of the average treatment effects (ATE) is a fundamental requirement for TTE
to yield valid biomedical conclusions} that guide clinical practice and inform evidence-based
decision making.

In this study, however, we show that ATE estimation for continuous outcomes in TTE can lead to nonnegligible \textit{bias} in the high-dimensional confounding settings common to RWE data when machine learning and artificial intelligence (ML/AI) regression is used to model potential outcomes. The root cause is likely systematic prediction bias (SPB), a widespread yet underappreciated property of ML/AI regression.

\textbf{Systematic Prediction Bias in ML/AI regression.}
Regression plays a central role in both statistics and machine learning as a standard framework for predicting continuous outcomes. Although AI has achieved remarkable progress in language and categorical tasks, such as text generation and topic classification, regression remains the foundational tool for applications involving continuous outcomes.
However, as illustrated in Figure~\ref{fig:spb} (a), standard ML/AI regression models, including random
forests, gradient boosting, kernel methods, deep neural networks, and regularized regression all exhibit a characteristic shrinkage pattern: large outcome values are
underpredicted while small outcome values are overpredicted, with predictions uniformly
shrinking toward the marginal outcome mean \citep{lee2025systematic}. 

This systematic prediction bias emerges
because commonly used objective functions such as mean squared error (MSE) favor solutions with
 bias in exchange for reduced variance, a trade-off that improves overall prediction
accuracy but introduces a deterministic, direction-dependent prediction error. While SPB may
reduce the MSE and is therefore ``acceptable'' in pure prediction settings, it is
not acceptable when ML/AI regression is used to compute potential outcomes in a causal
inference pipeline: the resulting ATE bias is structural and cannot be removed by
cross-fitting, sample splitting, ensemble aggregation, or any other standard ML practice, as demonstrated in Figure \ref{fig:spb} (b).
\begin{figure}[ht]
  \centering
  \includegraphics[width=1\linewidth]{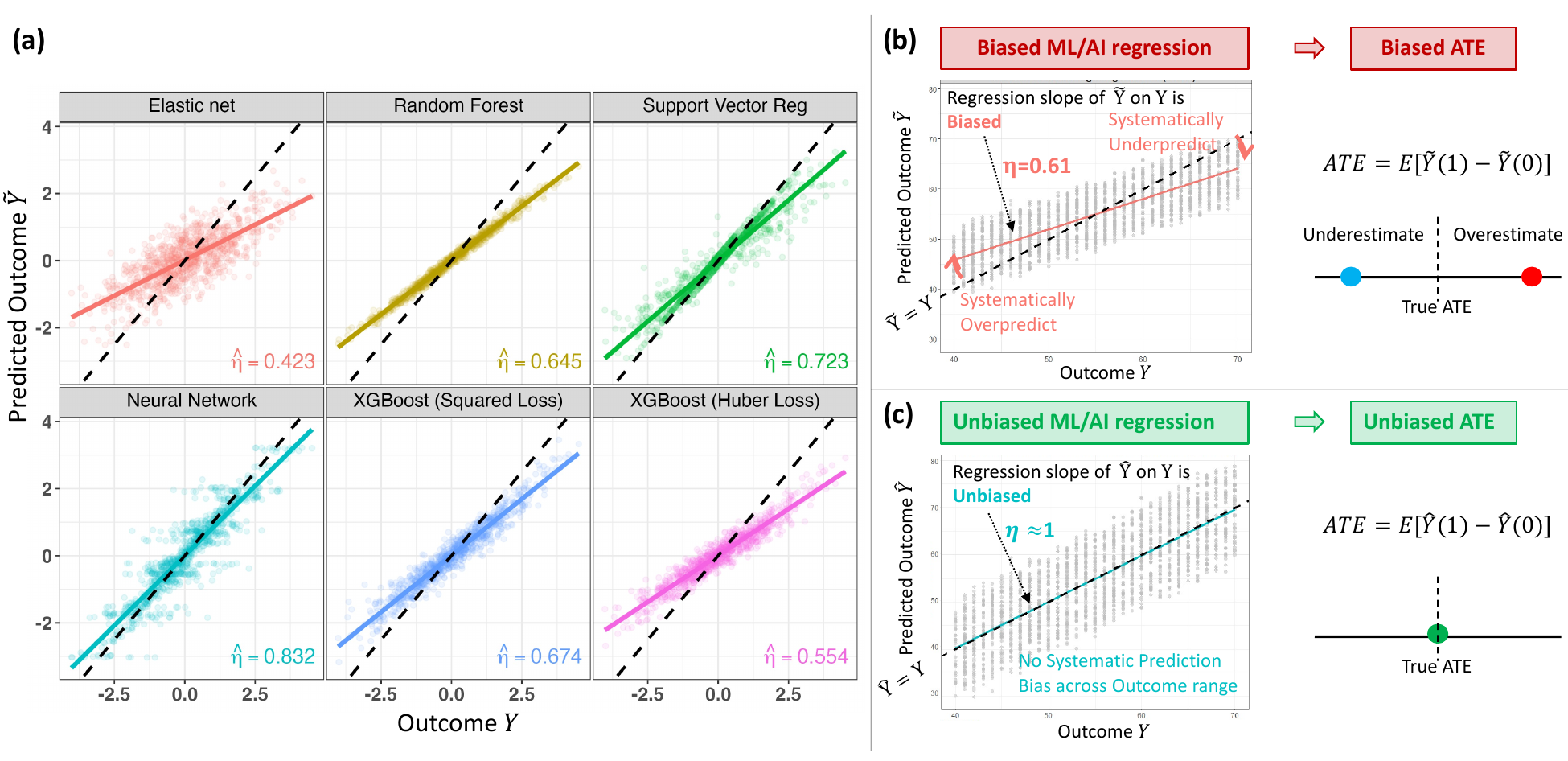}
  \caption{Illustration of systematic prediction bias  across ML/AI regression models
  and its propagation into causal inference.
  \textbf{(a)} Scatter plots of predicted outcome $\widehat{Y}$ vs.\ true outcome $Y$
  for representative ML/AI regression models (random forests, XGBoost, neural networks, LASSO, and XGBoost with Huber loss for both training and testing sets).
  The dashed diagonal is the no-bias reference ($\widehat{Y}=Y$); the solid line is the
  fitted regression with estimated slope $\widehat{\eta}$.
  SPB is present whenever $1-\widehat{\eta}\neq 0$, i.e., the solid line deviates from
  the dashed line. This pattern is ubiquitous across all ML regression (MLR) objective functions considered.
  \textbf{(b)} SPB in the MLR models propagates directly into biased ATE estimation; the magnitude and direction of ATE bias are governed by $\eta$.
  \textbf{(c)} The proposed unbiased ML regression (UMLR) framework achieves $1-\widehat{\eta}\approx 0$, yielding
  unbiased potential-outcome predictions and unbiased ATE estimates.}
  \label{fig:spb}
\end{figure}

\textbf{Contributions.}
In this article, we propose a novel causal inference framework using \textit{trustworthy} AI/ML regression
for RWD-based TTE. Specifically, we:

\begin{enumerate}
  \item identify a severe and previously unrecognized problem of biased ATE estimation in general AI/ML regression-based causal inference models, driven by the inherent systematic prediction bias in AI/ML regression;
  \item derive how the metric of systematic prediction bias  in AI/ML regression quantitatively determines the bias of ATE estimation in causal inference models;
  \item introduce a new causal inference framework built on unbiased ML regression (UMLR) that substantially mitigates ATE estimates;
  \item highlight, from a statistical standpoint, the importance of implementing \textit{trustworthy }AI/ML in causal inference, and provide a real-world application to UK Biobank data studying the effect of opioid therapy on blood pressure in patients with chronic pain.
\end{enumerate}

%-----------------------------------------------------------------------
\section{Methods}
\label{sec:methods}
%-----------------------------------------------------------------------

\subsection{Regression Prediction in Causal Inference}

RWD datasets routinely record hundreds to thousands of covariates per patient, necessitating flexible ML/AI models within causal inference pipelines. We focus on continuous outcomes, which encompass a broad range of clinically relevant primary endpoints, from blood pressure and cognitive function to disease progression scales  \citep{hubbard2024targettrial}.  Under the potential outcome
framework \citep{rubin1974estimating}, we denote the ATE as $\Delta = \mathbb{E}[Y(1) - Y(0)]$,
where $Y$ is a continuous outcome, $T \in \{0,1\}$ is a binary treatment indicator, $X \in
\mathbb{R}^p$ represents high-dimensional covariates, and $Y(t)$ is the potential outcome under
treatment $T = t$. We observe $n$ i.i.d.\ tuples $\{Y_i, X_i, T_i\}_{i=1}^n$. The identification of ATE requires several standard assumptions: consistency, the stable unit treatment value assumption (SUTVA), no unmeasured confounding, and positivity \citep{imbens2015causal}.

% \textcolor{red}{here do we need to also briefly state/review the standard assumptions such as no unmeasured confounding etc.~here for self-containedness? or are these conditions violated which results in the systematic bias in estimating $\tau(x)$? Good to clarify.} \tblue{Chixiang, do we need explicitly write down the assumptions or just a reference. }
% \tpurple{Identification of the ATE from observational data relies on the standard causal assumptions: (i) consistency, \(Y = Y(t)\) when $T=t$; (ii) unconfoundedness, \(\{Y(0), Y(1)\} \perp T \mid X\); and (iii) positivity, \(0 < e(X) < 1\), where \(e(X)=\Pr(T=1 \mid X)\) is the propensity score. Under these assumptions, $\Delta = \mathbb{E}\!\left[\mathbb{E}(Y \mid T=1, X) - \mathbb{E}(Y \mid T=0, X)\right].$}

The conditional average treatment effect (CATE) is $\tau(x) = \mu_1(x) - \mu_0(x)$, where
$\mu_t(x) = \mathbb{E}[Y \mid T=t, X=x]$. The ATE is then estimated as
\begin{equation}
  \widehat{\Delta} = \frac{1}{n}\sum_{i=1}^n \bigl[\hat{\mu}_1(X_i) - \hat{\mu}_0(X_i)\bigr].
  \label{eq:ate}
\end{equation}
% Doubly robust estimators such as AIPW augment the outcome regression with inverse probability
% weighting based on the propensity score $e(X) = P(T=1 \mid X)$:
% \begin{equation}
%   \widehat{\Delta}^{\mathrm{DR}} = \frac{1}{n}\sum_{i=1}^n
%   \left[\hat{\mu}_1(X_i) - \hat{\mu}_0(X_i)
%     + \frac{T_i}{e(X_i)}\bigl(Y_i - \hat{\mu}_1(X_i)\bigr)
%     - \frac{1-T_i}{1-e(X_i)}\bigl(Y_i - \hat{\mu}_0(X_i)\bigr)
%   \right].
%   \label{eq:aipw}
% \end{equation}
Unbiased prediction of the potential outcome $\hat{\mu}_t$ yields an ATE estimate approximating the RCT benchmark; however, systematic bias in $\hat{\mu}_t$ propagates directly into the ATE estimate. We next articulate how systematic bias in ML/AI regression induces biased ATE estimates.
%across commonly used causal inference frameworks.

\subsection{Systematic Prediction Bias (SPB) in ML/AI Regression}
\label{sec:spb}

\textbf{ML/AI regression notation.}
Let $f$ denote an ML/AI predictive function $f$ that predicts the continuous outcome $\widehat{Y}_i = f(X_i)$ thus providing $\hat{\mu}_t(X_i)$ for all treatment arms. Without loss of generality, we center the outcome so that $n^{-1}\sum_{i=1}^n Y_i=0$.
% center the outcome so that $n^{-1}\sum_{i=1}^n Y_i=0$.
% Let $X_i \in \mathbb{R}^p$ denote the vector of $p$ high-dimensional predictors for
% subject $i$, and let $Y_i$ denote the corresponding continuous outcome, for $i=1,\ldots,n$.
% The training data are $\mathbf{X}=(X_1,\ldots,X_n)^\top\in\mathbb{R}^{n\times
% p}$ and $\mathbf{Y}=(Y_1,\ldots,Y_n)^\top\in\mathbb{R}^n$. \textcolor{red}{is this repeated from the notation introduced in previous section? or better to use a unified notation system cover both sections to avoid ambiguity.} Without loss of generality, we
% center the outcome so that $n^{-1}\sum_{i=1}^n Y_i=0$. Let $f$ denote an ML/AI predictive function $f$ \textcolor{red}{maybe use $\hat f$ to highlight it is an estimator?}, and define $\widehat{Y}_i = f(X_i)$ as the fitted value, which provides estimates $\hat{\mu}_t(X_i)$.

\textbf{Systematic prediction bias of ML/AI regression.}
We consider predicted outcome as unbiased given $Y$, if $\mathbb{E}(\widehat{Y}_i \mid Y_i=y^*)=y^*$ for
all $y^*$ in the support of $Y$ (see Figure \ref{fig:spb} (c)). Throughout the paper, we use $\widehat{Y}$ to denote a generic prediction and
$\widetilde{Y}$ to denote the special case of $\widehat{Y}$ that is systematically biased;
that is, $\mathbb{E}(\widetilde{Y}_i \mid Y_i=y^*)\neq y^*$. The prediction
bias is systematic if it depends on the true outcome through a deterministic function
$h(\cdot)$, for example, $\mathbb{E}(\widetilde{Y}_i -  Y_i\mid Y_i=y^*)  = h(y^*).$ Systematic prediction bias can be illustrated using scatter plots of $\hat{Y}$ versus $Y$ (e.g., Figure \ref{fig:spb} (b)).
% \begin{equation}
%   \mathbb{E}(\widetilde{Y}_i \mid Y_i=y^*) - y^* = h(y^*).
%   \label{eq:spb-general}
% \end{equation}

 As demonstrated in Figure \ref{fig:spb}(a), SPB is present in all tested ML/AI regression models for both the training and testing datasets \citep{BELITZ:2021105006,Zhang:2012_RF:Bias,lee2025systematic}. The term “systematic” indicates that the prediction error is not merely random noise, but instead follows a non-random pattern determined by the outcome value itself. \cite{lee2025systematic} shows that predicted outcomes exhibiting SPB can reduce empirical loss functions, such as the mean squared error (MSE). Therefore, most regression models tend to favor biased predictions due to the well-known bias-variance tradeoff.
SPB is not due to the use of optimization algorithms in ML/AI regression. Instead, it is a universal phenomenon arising from the use of objective functions for continuous outcome like MSE, and is also present in ordinary least squares (OLS) regression. Although OLS yields unbiased parameter estimates and unbiased predictions conditional on predictors, $\mathbb{E}(\widehat{Y} - Y \mid X) = 0$, the underappreciated fact is that $\mathbb{E}(\widehat{Y} - Y \mid Y) \neq 0$ \citep{Treder:2021}.  %Therefore, causal inference in observational studies with OLS-derived potential outcomes can also yield biased ATE estimates. %However, the bias of estimated ATE is more pronounced in high-dimensional settings.

\textbf{A Metric for Characterizing Systematic Prediction Bias.}
As the degree of systematic prediction bias varies across $Y$, an overall metric is desirable to characterize it over the support of $Y$. Based on a large body of literature \citep{SMITH:2019,Butler:2021}, the most frequently observed form of $h(y)$ follows a linear pattern:
\begin{equation}
  \mathbb{E}(\widetilde{Y}_i - Y_i\mid Y_i=y^*)=h(y^*) - y^* = (\eta-1)\,y^*, \quad \eta \in [0,1].
  \label{eq:spb-linear}
\end{equation}
That is, $\mathbb{E}(\widetilde{Y}_i \mid Y_i=y^*)=\eta\,y^*$: every conditional expectation of the predicted outcome is shrunk toward zero (i.e., toward the grand mean of $Y$ before centering) by a factor of $\eta$. Under this linear trend pattern, the bias is most pronounced at the extremes: large outcome values are systematically underpredicted and small outcome values are systematically overpredicted (see Figure~\ref{fig:spb} (a)). We therefore quantify the degree of SPB by $1-\eta$ (or equivalently $100(1-\eta)\%$), where $1-\eta=0$ indicates no bias. In practice, $\eta$ can be estimated by the OLS slope $\widehat{\eta}$ obtained from regressing ML/AI regression predicted outcome $\widehat{Y}$ on $Y$.

\subsection{Biased ATE Estimation Induced by Systematic Prediction Bias}
\label{sec:biased-ate}

Let $\mu_t^*(X) = \mathbb{E}[Y(t) \mid X]$ denote the `oracle' conditional mean potential outcome
under treatment $t \in \{0,1\}$. Define $\bar{\mu}_t^{(s)} = \mathbb{E}[\mu_t^*(X) \mid T=s]$ as
the mean of the true potential outcome surface over the stratum $T=s$. Specifically, the subscript $t$ indicates the potential treatment group while the subscript $s$ indicates the assigned treatment group.

When classic ML/AI regression-based outcome model $\hat{\mu}_1(X)=f(X)$  is trained on the treated subsample
$\{({X}_i, Y_i) \mid T_i = 1\}_{i=1}^n$. Due to SPB, the expected prediction for a unit drawn
from the treated arm, given the true conditional mean $\mu_1^*(X)$, takes the form
%\begin{equation}
 %   \mathbb{E}\bigl[\hat{\mu}_1(X) \mid \mu_1^*(X),\, T=1\bigr]
 %   = \mathbb{E} \bigl[ \eta_1^{(1)}\,\mu_1^*(X) + %\bigl(1 - \eta_1^{(1)}\bigr)\,\bar{\mu}_1^{(1)} %\bigr]=\bar{\mu}_1^{(1)},
 %   \qquad X \sim P(X \mid T = 1),
 %   %\label{eq:mu1}
%\end{equation}
\begin{equation*}
    \mathbb{E}\bigl[\hat{\mu}_1(X) \mid \mu_1^*(X),\, T=1\bigr]
    \;=\; \eta_1^{(1)}\,\mu_1^*(X) + \bigl(1 - \eta_1^{(1)}\bigr)\,\bar{\mu}_1^{(1)},
    \qquad X \sim P(X \mid T = 1),
\end{equation*}
where $\eta_1^{(1)} \in [0,1]$ governs the degree of shrinkage toward the in-distribution
mean $\bar{\mu}_1^{(1)} := \mathbb{E}[\mu_1^*(X) \mid T=1]$. Taking expectations over
$X \sim P(X \mid T=1)$:
\begin{equation}
    \mathbb{E}_{P(X\mid T=1)}\bigl[\hat{\mu}_1(X)\bigr]
    \;=\; \eta_1^{(1)}\,\bar{\mu}_1^{(1)} + \bigl(1-\eta_1^{(1)}\bigr)\,\bar{\mu}_1^{(1)}
    \;=\; \bar{\mu}_1^{(1)}.
    \label{eq:mu1_in}
\end{equation}

%where $\eta_1^{(1)} \in [0,1]$ is the in-distribution %slope parameter and
%$\bar{\mu}_1^{(1)} = \mathbb{E}[\mu_1^*(X) \mid T=1]$ is %the training-set mean.
Now consider the predicted potential outcome $\hat{\mu}_1(X)$ evaluated on a unit
drawn from the control arm, $X \sim P(X \mid T=0)$. This is the out-of-distribution
(OOD) setting for $\hat{\mu}_1$, since i) the model was trained exclusively on
$\{( {X}_i, Y_i) \mid T_i = 1\}_{i=1}^n$ and is now applied to covariate
regions underrepresented or absent in its training support, and ii) $p$ is high-dimensional. Under the ML/AI
regression model with systematic prediction bias (SPB), the expected prediction
given the true conditional mean $\mu_1^*(X)$ takes the form
%\begin{equation}
%    \mathbb{E}\bigl[\hat{\mu}_1(X) \mid \mu_1^*(X),\, %T=0\bigr]
%    = \eta_1^{(0)}\,\mu_1^*(X) + \bigl(1-%\eta_1^{(0)}\bigr)\,\bar{\mu}_1^{(\alpha)},
%    \qquad X \sim P(X \mid T=0),
%    \label{eq:mu0}
%\end{equation}
%where $\eta_1^{(0)} \in [0,1]$ is the out-of-distribution %slope parameter and
%the shrinkage target remains $\bar{\mu}_1^{(1)} = %\mathbb{E}[\mu_1^*(X) \mid T=1]$,
%the training-set mean of the treated arm. 
\begin{equation*}
    \mathbb{E}\bigl[\hat{\mu}_1(X) \mid \mu_1^*(X),\, T=0\bigr]
    \;=\; \eta_1^{(0)}\,\mu_1^*(X) + \bigl(1-\eta_1^{(0)}\bigr)\,\bar{\mu}_1^{({\alpha}_1)},
    \qquad X \sim P(X \mid T=0),
\end{equation*}
so that taking expectations over $X \sim P(X \mid T=0)$:
\begin{equation}
    \mathbb{E}_{P(X\mid T=0)}\bigl[\hat{\mu}_1(X)\bigr]
    \;=\; \eta_1^{(0)}\,\bar{\mu}_1^{(0)} + \bigl(1-\eta_1^{(0)}\bigr)\,\bar{\mu}_1^{({\alpha}_1)},
    \label{eq:mu1_ood}
\end{equation}
where $\eta_1^{(0)} \in [0,1]$ is the out-of-distribution slope parameter and
the shrinkage target remains $\bar{\mu}_1^{(1)} = \mathbb{E}[\mu_1^*(X) \mid T=1]$,
the training-set mean of the treated arm, while $\bar{\mu}_1^{(0)} := \mathbb{E}[\mu_1^*(X) \mid T=0]$.

The model has sparse coverage of $P(X \mid T=0)$ and therefore collapses more
aggressively toward $\bar{\mu}_1^{(1)}$ for control-arm units than for treated-arm
units. Here, we parametrize the shrinkage target as $\bar{\mu}_1^{({\alpha}_1)} = w\,\bar{\mu}_1^{(1)} + (1-w)\,\bar{\mu}_1^{(0)}$
for some $w \in (0,1)$, where $w$ reflects the degree to which the model extrapolates
toward the training-set mean of the treated arm.

\begin{theorem}
For each treatment arm $t \in \{0,1\}$, let $\eta_t^{(t)} \in [0,1]$
denote the in-distribution slope parameter and $\eta_t^{(1-t)} \in [0,1]$ denote the
out-of-distribution slope parameter evaluated on the opposite arm. Let $\pi = \mbox{P}(T=1)$. Then the bias of the
estimated ATE via the outcome regression estimators \eqref{eq:mu1_in} and \eqref{eq:mu1_ood} is
\begin{align}
    \mathrm{Bias}_{\mathrm{SPB}}
    &\;:=\;
    \Bigl(\mathbb{E}_{P(X)}\bigl[\hat{\mu}_1(X)\bigr] - \mathbb{E}_{P(X)}\bigl[\mu_1^*(X)\bigr]\Bigr)
    \;-\;
    \Bigl(\mathbb{E}_{P(X)}\bigl[\hat{\mu}_0(X)\bigr] - \mathbb{E}_{P(X)}\bigl[\mu_0^*(X)\bigr]\Bigr)
    \nonumber\\[6pt]
    &\;=\;
    (1-\pi)\bigl(1-\eta_1^{(0)}\bigr)\,w_1\!\left(\bar{\mu}_1^{(1)}-\bar{\mu}_1^{(0)}\right)
    \;+\;
    \pi\bigl(1-\eta_0^{(1)}\bigr)\,w_0\!\left(\bar{\mu}_0^{(1)}-\bar{\mu}_0^{(0)}\right).
\end{align}
\label{thm:bias_spb}
\end{theorem}

\subsection{From Unbiased ML/AI Regression to Unbiased Causal Inference}
\label{sec:umlr}
\textbf{Unbiased ML/AI Regression via Constrained Objective Functions.}
We adopt the mean-anchoring constrained objective function to eliminate SPB at the model-training stage \citep{lee2025systematic}. Specifically, we have
:
\begin{equation}
  \min_{f}\;\sum_{i=1}^{n} \ell\bigl(f(X_i), Y_i\bigr) + \lambda\,\Omega(f)
  \quad\text{s.t.}\quad
  \sum_{i \in \mathcal{R}_1}\bigl(f(X_i) - Y_i\bigr) = 0
  \;\;\text{and}\;\;
  \sum_{i \in \mathcal{R}_2}\bigl(f(X_i) - Y_i\bigr) = 0,
  \label{eq:obj}
\end{equation}

where $\Omega$ is a regularizer with tuning parameter $\lambda$. The two index sets
$\mathcal{R}_1$ and $\mathcal{R}_2$ are disjoint and complementary, partitioning the
sample of $Y$, with $\mathcal{R}_1 \cup \mathcal{R}_2 = \{1,\ldots,n\}$. For example,
using the sample mean as the cut-off gives $\mathcal{R}_1 = \{i : Y_i \leq \bar{Y}\}$
and $\mathcal{R}_2 = \{i : Y_i > \bar{Y}\}$. Imposing the two mean-anchoring constraints
enforces $\eta = 1$, i.e., $\mathbb{E}(\widetilde{Y}_i \mid Y_i = y^*) = y^*$ for all
$y^*$, eliminating SPB across all ML/AI regression model classes, including tree-based
algorithms, kernel methods, and regularized regression \citep{lee2025systematic}.

Let $\hat{\mu}_t^{U}(X)$ denote the potential outcome estimated by a UMLR model
trained under the constrained objective function \eqref{eq:obj}.  Then, we  define the UMLR based   ATE   as
\begin{equation}
  \widehat{\Delta}^U = \frac{1}{n}\sum_{i=1}^n \bigl[\hat{\mu}^U_1(X_i) - \hat{\mu}^U_0(X_i)\bigr].
  \label{eq:ate}
\end{equation}

\textbf{Unbiased ML regression-based Causal Inference Model.} Because SPB is eliminated at the model-fitting stage, yielding $\eta=1$ for every ML/AI regression model, the systematic bias terms derived in Section~\ref{sec:biased-ate} vanish exactly. Consequently, the resulting ATE estimator is unbiased regardless of the choice of causal inference model, as established in Theorem~\ref{thm:main}.

% \tred{[HL:] in theorem below using $f^\ast$ can make confusion, can we use $\hat{\mu}$}
% \begin{theorem} 
% \label{thm:main}  
% Suppose UMLR $f^*$ is the solution to \eqref{eq:obj} satisfying the two constraints. Then the ATE estimator constructed by UMLR  is unbiased, $ \mathbb{E}(\widehat{\Delta}^U - \Delta)=0$ .  
% \end{theorem}

\begin{theorem} 
\label{thm:main}  
For each level of treatment $T=\{0,1\}$, the outcome learner $\hat{\mu}_t^U$ is trained under the constrained objective function
\eqref{eq:obj}. Then the UMLR based ATE estimate in \eqref{eq:ate} is unbiased: $\mathbb{E}(\widehat{\Delta}^U - \Delta)=0$.
\end{theorem}

% The proof follows directly from Theorem~1 in \citet{lee2025systematic}: the two mean-anchoring
% constraints together enforce $\eta = 1$, removing the shrinkage that underlies SPB.

% \textbf{Summary.}
% The unbiasedness results establish a unified guarantee: replacing the ML/AI outcome regression component in a causal inference pipeline with UMLR models trained under the constrained objective in  \eqref{eq:obj} yields an unbiased ATE estimator. This bias suppression is not model-specific, but follows directly from eliminating SPB in outcome prediction. Consequently, the result applies broadly to S-, T-,   X-learners,  doubly robust estimators, and others. The application of UMLR to causal inference further underscores the importance of \textit{trustworthy} AI/ML for RWD analysis.

%-----------------------------------------------------------------------
\section{Simulation Studies}
\label{sec:sim}
%-----------------------------------------------------------------------

We conducted simulation studies to assess the finite-sample performance of ML/AI regression (MLR) and UMLR and to quantify the impact of systematic prediction bias on causal effect estimation.

\textbf{Setup.} The covariate matrix $X \in \mathbb{R}^{n \times p}$ was generated independently from a standard normal distribution, $X_{ij} \overset{\text{iid}}{\sim} \mathcal{N}(0, 1)$, for $i = 1, \ldots, n$ and $j = 1, \ldots, p$. Treatment $T_i \in \{0, 1\}$ wss assigned according to a logistic propensity model,
\begin{equation*}
    T_i \mid X_i \sim \text{Bernoulli}\big(e(X_i)\big), \quad e(X_i) = \sigma(X_i^\top \boldsymbol{\gamma}),
\end{equation*}
where $\sigma(\cdot)$ denotes the logistic function and $\boldsymbol{\gamma} \in \mathbb{R}^p$ is a sparse vector with $s$ active components. The potential outcomes were generated as
\begin{equation*}
    Y_i(1) = \mu_1 + X_i^\top \boldsymbol{\beta}_1 + \varepsilon_i, \quad Y_i(0) = \mu_0 + X_i^\top \boldsymbol{\beta}_0 + \varepsilon_i,
\end{equation*}
where $\varepsilon_i \overset{\text{iid}}{\sim} \mathcal{N}(0,1)$, and $\boldsymbol{\beta}_1, \boldsymbol{\beta}_0 \in \mathbb{R}^p$ are sparse coefficient vectors each with $s$ active components. The observed outcome is $Y_i = T_i Y_i(1) + (1 - T_i)Y_i(0)$. The individual-level conditional average treatment effect (CATE) is defined as
\begin{equation*}
    \tau(X_i) = Y_i(1) - Y_i(0),
\end{equation*}
and the true ATE in each replication is computed as
\begin{equation*}
\Delta
=
\mathbb{E}\!\left[\tau(X)\right]
=
\frac{1}{n}\sum_{i=1}^{n}\tau(X_i).
\end{equation*}
We consider three settings varying $n$ and $p$:
$(n, p) \in \{(500, 200),\ (1{,}000, 200),\ (1{,}000, 500)\}$.

% \tpurple{We repeated each experiment over 200 Monte Carlo replications. Three $(n, p)$ configurations were considered: $(500, 200)$ as the baseline setting, $(500, 500)$ to assess higher dimensionality, and $(1000, 200)$ to evaluate larger sample sizes. We further generated data under varying levels of SBP, induced through shrinkage parameters $\eta \in (0,1)$. We implement S-, T-, and X-learners using both MLR and the proposed UMLR as outcome models. We evaluated estimator performance across different methods, using the Randomized Controlled Trial (RCT) estimator as a gold-standard benchmark. For each scenario, we assessed the bias, standard deviation, and coverage probability of the estimated ATE across varying degrees of shrinkage.}

%\paragraph{SPB Scenarios.}
%We consider two shrinkage configurations to examine how prediction bias propagates through each learner, and how systematic shrinkage leads to underestimation or overestimation of the ATE.

%\begin{itemize}
%    \item \textbf{Uniform shrinkage.} We first consider a uniform shrinkage setting in which
%    \[
%    \eta_s = \eta_{t1} = \eta_{t0} = \eta_{x1} = \eta_{x0} = 0.6,
%    \]
%    so that all learners are equally affected by prediction bias.

%    \item \textbf{Heterogeneous shrinkage.} Shrinkage varies across treatment arms and estimation stages,
 %   \[
 %   \eta_s = 0.6, \quad \eta_{t1} = 0.8, \quad \eta_{t0} = 0.6, \quad \eta_{x1} = 0.8, \quad \eta_{x0} = 0.6,
 %   \]
 %   so that the arms and the X-learner second-stage models face different shrinkage, inducing asymmetric bias across learners.
%\end{itemize}
 
 \textbf{Results.} The simulation results are summarized in Table~\ref{tab:table1}. Across all settings, the UMLR-based causal effect estimator outperforms MLR-based causal inference methods by substantially mitigating systematic prediction bias. More importantly, the 95\% confidence interval coverage of the MLR-based methods approaches zero, creating a serious risk of misleading real-world evidence-based decisions. In contrast, the 95\% confidence interval coverage of the UMLR-based methods approaches one.

\begin{table}[H]
\centering
\caption{%
Simulation results comparing the performance of MLR and UMLR-based causal inference under systematic prediction bias.  }
\label{tab:table1}
\begin{threeparttable}
\setlength{\tabcolsep}{10pt}
\renewcommand{\arraystretch}{1.22}
\begin{tabular}{l ccc ccc ccc}
  \toprule
  & \multicolumn{3}{c}{$n{=}500,\;p{=}200$}
  & \multicolumn{3}{c}{$n{=}1{,}000,\;p{=}200$}
  & \multicolumn{3}{c}{$n{=}1{,}000,\;p{=}500$} \\
  \cmidrule(lr){2-4}\cmidrule(lr){5-7}\cmidrule(lr){8-10}
  Method 
    & \makecell{Bias (\%)} & RMSE & Cov.
    & \makecell{Bias (\%)} & RMSE & Cov.
    & \makecell{Bias (\%)} & RMSE & Cov. \\
  \midrule
  MLR 
    & -19.96 & 2.96 & 0.02 
    & -12.02 & 1.97  & 0.01          
    & -22.19 & 3.25  & 0.00  \\
  UMLR 
    & -4.59  & 1.93 & 0.81
    & -0.82  & 1.38 & 1.00       
    & -4.42  & 1.99 & 0.71 \\
  \bottomrule
\end{tabular}
\begin{tablenotes}[flushleft]\footnotesize
  \item Bias(\%): mean absolute bias as \% of true ATE; Cov.: empirical 95\% CI coverage.
\end{tablenotes}
\end{threeparttable}
\end{table}

Figure~\ref{fig:fig2} further illustrates the differences in ATE estimation and SPB across regression models. Figure~\ref{fig:fig2}(b) explains the bias discrepancy observed in Figure~\ref{fig:fig2}(a): under MLR, the estimated SPB slopes are $\hat{\eta}_{1}^{(0)} = 0.51$ and $\hat{\eta}_{0}^{(1)} = 0.53$ for the $Y(1)$ and $Y(0)$ counterfactual arms, respectively. These values indicate that the model substantially shrinks its predictions toward the training-set mean when predicting potential outcomes for the opposite treatment arm. UMLR largely mitigates this shrinkage, increasing the out-of-distribution slopes to $\hat{\eta}_{1}^{(0)} = 0.85$ and $\hat{\eta}_{0}^{(1)} = 0.92$, thereby approaching the unbiased ideal of $\eta_{t}^{(s)} = 1$.

These findings are consistent with the theoretical results established in Theorems~\ref{thm:bias_spb} and~\ref{thm:main}, which show that UMLR substantially reduces SPB in outcome prediction, thereby mitigating ATE estimation bias and improving the accuracy of statistical inference.

% {\color{purple}
% Table~\ref{tab:table1} summarizes performance across varying sample sizes $n$ and dimensionality $p$. Overall, UMLR   outperforms MLR for all settings, achieving smaller bias and better coverage.}
% {\color{purple}
% \textbf{Results: } 
% Figure~\ref{fig:fig2}(a) shows that MLR produces a noticeably biased
% ATE estimator, whereas UMLR substantially reduces the bias, with its distribution centered much closer to zero. Figure~\ref{fig:fig2}(b) further explains the bias reduction observed in Figure~\ref{fig:fig2}(a). Under MLR, the OOD slope parameters are $\hat{\eta}_{1}^{(0)} = 0.51$ and
% $\hat{\eta}_{0}^{(1)} = 0.53$ for the $Y(1)$ and $Y(0)$ counterfactual arms
% respectively, indicating that the model shrinks its predictions substantially
% toward the training-set mean when extrapolating to the opposite arm. UMLR
% largely mitigates this shrinkage, increasing the OOD slopes to $\hat{\eta}_{1}^{(0)} = 0.85$
% and $\hat{\eta}_{0}^{(1)} = 0.92$, closer to the ideal $\eta_{t}^{(s)} = 1$. These results demonstrate that the proposed UMLR framework effectively mitigates systematic prediction bias in ATE estimation.
% }
\begin{figure}[htbp]
    \centering
    \includegraphics[width=1\linewidth]{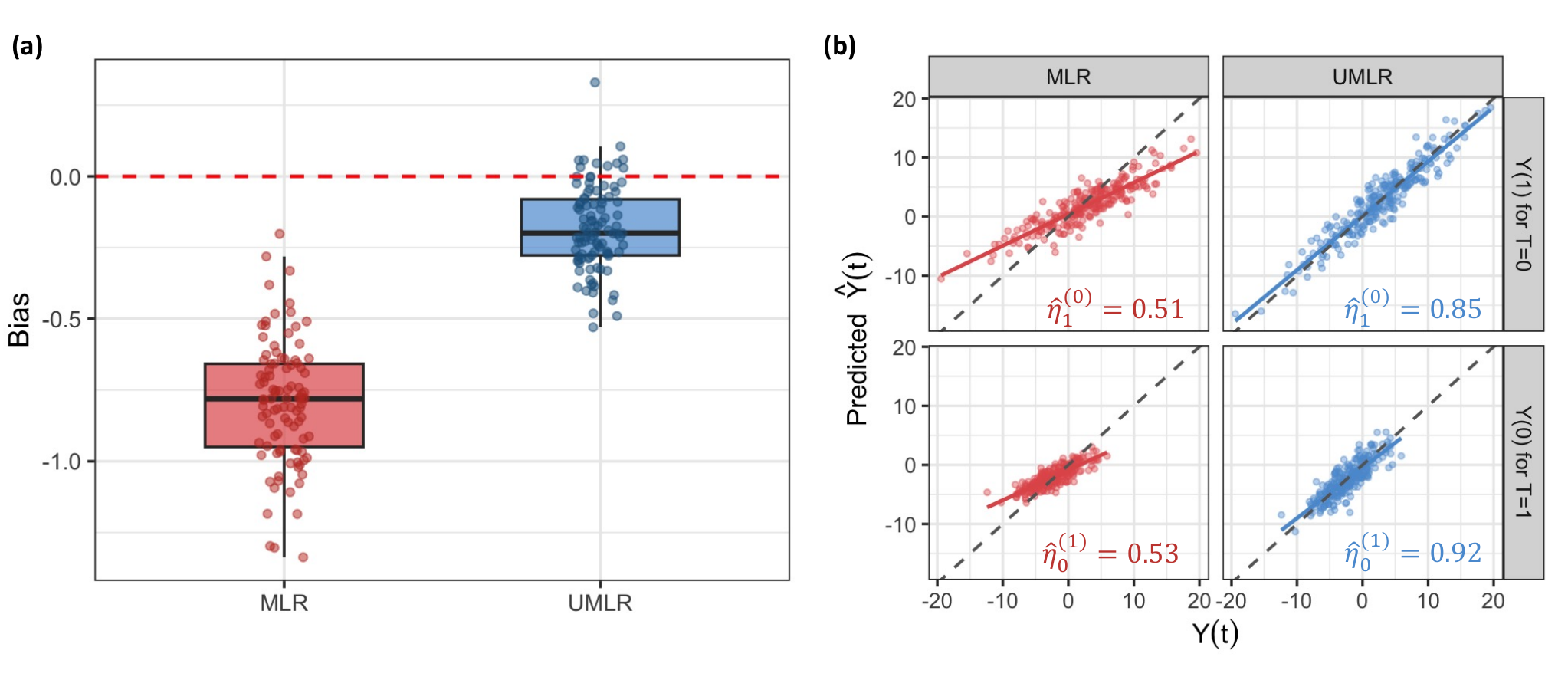}
\caption{Bias in estimated ATE for MLR and UMLR ($n = 500$, $p = 200$):
\textbf{(a)} Distribution of ATE bias across 100 Monte Carlo simulations.
\textbf{(b)} Scatter plots for $\hat{\eta}_{1}^{(0)}$ and
$\hat{\eta}_{0}^{(1)}$ by MLR (red) and UMLR (blue).}
    \label{fig:fig2}
\end{figure}

%-----------------------------------------------------------------------
\section{RWD Application: Opioid Therapy and Blood Pressure in Chronic Pain }
\label{sec:app}
%-----------------------------------------------------------------------

%We applied the proposed framework to study the ATE of opioid use on systolic blood pressure (SBP) in patients with chronic pain using data from the UK Biobank. The UK Biobank is a large-scale prospective cohort comprising approximately 500,000 participants aged 40--69 at recruitment, with longitudinal follow-up including hospital episode statistics, primary care records, and repeated assessment-center measurements. Patients with a diagnosis of chronic pain and complete baseline covariate data were included, resulting in a cohort of $n = \text{19,736}$ participants, of whom $523$ received opioid therapy (treatment) and $19{,}213$ did not receive any medication (control).

%Covariates included demographics, comorbidities, concurrent medications, baseline blood test biomarkers, and socioeconomic indices, with a total of $p = 35$ features (see Appendix). The ATE was estimated using MLR and UMLR embedded within S-, T-, and X-learners. Propensity Score Matching (PSM) was additionally included as a benchmark estimator.

We applied the proposed framework to estimate the ATE of opioid use on systolic blood
pressure among participants with chronic pain using UK Biobank data. Chronic pain was
defined as self-reported pain lasting more than three months in any of seven body regions
(headaches, facial pain, neck/shoulder, back, stomach/abdominal, hip, and knee) at
baseline. Opioid use was determined from self-reported current medications at the imaging
visit, classified using Anatomical Therapeutic Chemical (ATC) codes
\citep{gao2024association}. The primary outcome was the mean of two systolic blood
pressure measurements at the imaging visit, with fewer than two measurements treated as
missing. The final analytic sample comprised 19,736 participants with chronic pain, of
whom 523 were opioid users and 19,213 were non-users. Covariates included age, sex, 
baseline blood biomarkers and others, yielding $p = 35$ variables in total (see Appendix~Table~\ref{tab:descriptive}). The ATE
was estimated using MLR and UMLR incorporated into S-, T-, and X-learners \citep{kunzel2019metalearners}.
  Propensity score matching (PSM) was also applied as a benchmark, because it does not rely on ML/AI regression-predicted outcomes and is therefore not affected by SPB \citep{rosenbaum1983central}. Following matching, confounding variables were well balanced between the two treatment arms (see Appendix~A.4), and therefore the average treatment effect on the treated (ATT) estimated by PSM was used as the reference, while noting that the ATT estimand differs from the ATE and the sample size for PSM is smaller. 

 {As demonstrated in Figure~\ref{fig:fig3}(a), both UMLR-based and conventional
MLR-based causal inference models indicate that opioid treatment for pain management
significantly reduces systolic blood pressure relative to no medication (all $p < 0.001$),
consistent with previous studies \citep{sacco2013relationship, qin2026chronic}. However,
the estimated ATEs differ substantially: UMLR-based estimators yield larger negative
estimates of $-4.72$, $-4.94$, and $-5.29$ mmHg for the S-, T-, and X-learners,
respectively, closely aligned with the PSM benchmark of $-5.08$ mmHg, whereas
MLR-based estimators produce substantially attenuated estimates of $-2.41$, $-3.83$,
and $-2.36$ mmHg. This attenuation is consistent with the SPB-induced bias
characterized in Section~\ref{sec:biased-ate}. As shown in Figure~\ref{fig:fig3}(b), conventional MLR exhibits substantial SPB, with $\widehat{\eta}_S = 0.409$ and $\widehat{\eta}_T = 0.410$ on the training set, and $\widehat{\eta}_S = 0.148$ and $\widehat{\eta}_T = 0.147$ on the test set. In contrast, UMLR achieves near-zero SPB, with $\widehat{\eta}_S = 0.922$ and $\widehat{\eta}_T = 0.923$ on the training set, and $\widehat{\eta}_S = 0.903$ and $\widehat{\eta}_T = 0.904$ on the test set. The discrepancy in ATE estimates between MLR- and UMLR-based causal inference models is therefore likely driven by SPB-induced bias. 

\begin{figure}[ht]
  \centering
  \includegraphics[width=1\linewidth]{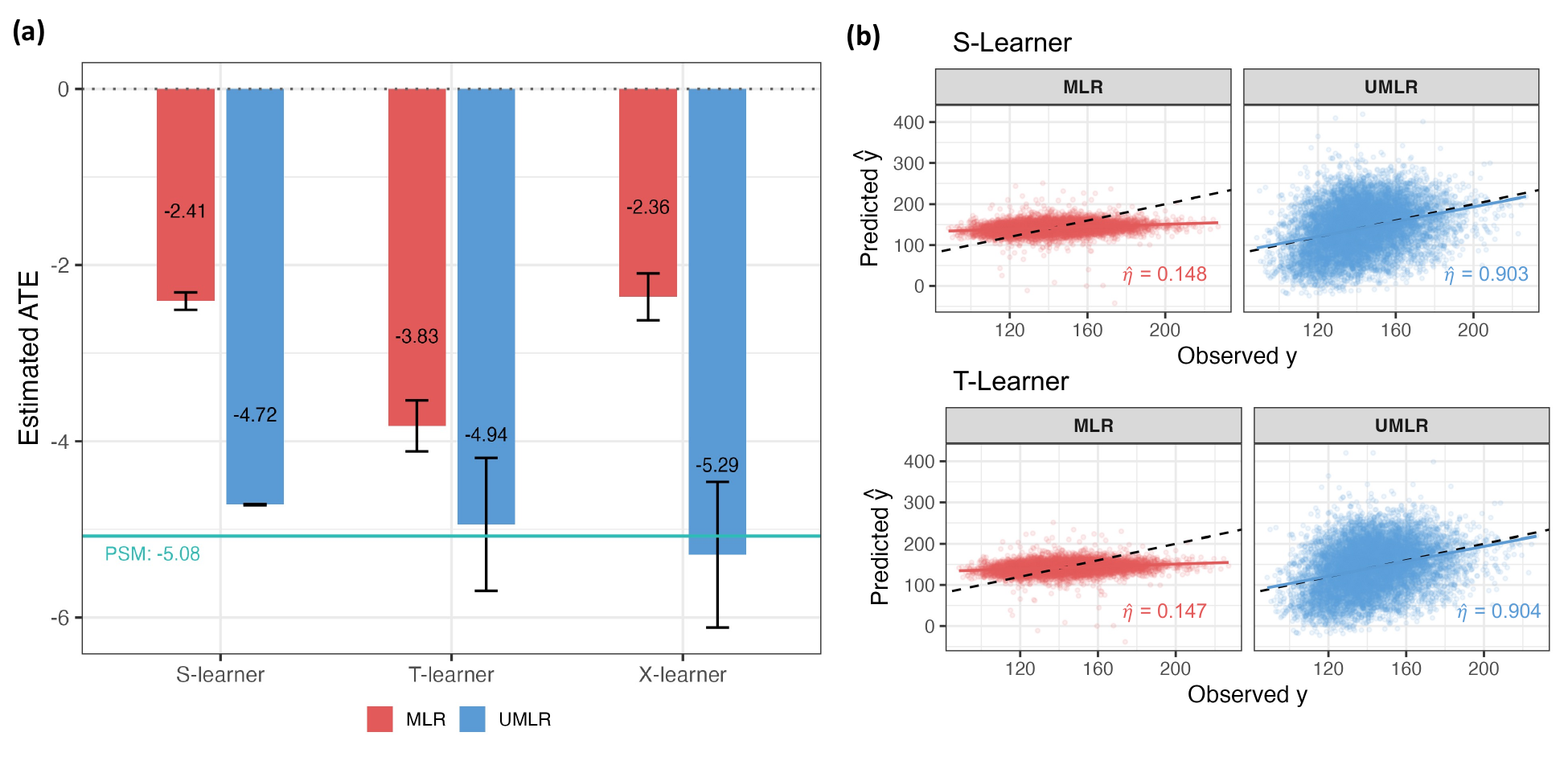}
  \caption{Results of real data application on mean systolic blood
pressure under opioid treatment vs. no-medication.
  \textbf{(a)} ATE estimates with 95\% confidence intervals. %\tbrown{there is very narrow confidence interval for UMLR for S-learner. Double check if there is any error. also, we need to be careful about CI for our method, where the asymptotic theory might be different from existing CATE. better use bootstrap as the way to do inference to avoid critique} across MLR- and UMLR-based meta-learners. Confidence intervals are constructed from the CATE, with the PSM estimate as a reference line.  
  \textbf{(b)} Predicted vs. observed systolic blood
pressure for the S-learner and T-Learner under MLR and UMLR for the testing datasets. The dashed line denotes the prediction without systematic bias. (i.e.,  $\eta=1$).
%\tbrown{we need to mention whether (b) is evaluated at the test dataset (i guess so but important to mention)}
} 
  \label{fig:fig3}
\end{figure}
%-----------------------------------------------------------------------
\section{Discussion}
\label{sec:discussion}
%-----------------------------------------------------------------------
 Despite the advancement in ML/AI, a structural vulnerability still remains in current ML/AI regression tasks and its application in RWE-based TTE analysis: systematic prediction bias, widely seen in standard ML/AI regression algorithms, propagates into biased ATE estimates for regression-based causal machine learning methods in ways that may not be corrected by cross-fitting, ensemble methods, or model selection. We have proposed a principled remedy in the form of UMLR, which mitigate SPB at the model-fitting stage. It has been shown that the resulting UMLR-based causal inference framework yields unbiased ATE estimates when SPB is completely corrected. In practice, UMLR often achieves $\hat{\eta}$ values close to 1, such as $\sim 0.9$, compared with $\hat{\eta} < 0.5$ under MLR. Therefore, even when SPB is not fully eliminated, UMLR can substantially reduce ATE estimation error.

Notably, outcome regression-based causal inference models are attractive in many applications. Compared with weighing-based approaches, outcome regression-based causal inference models can be less sensitive to instability arising from extreme weights, which may occur in settings with limited overlap (as in our UKB data application). In many practical scenarios, this can translate into improved finite-sample performance, including MSE, particularly when the outcome model is reasonably well-specified. Outcome regression-based models are especially convenient in settings with continuous treatments or exposures, where direct modeling of the outcome surface can be more straightforward than constructing and stabilizing generalized propensity scores.

Based on these findings, we recommend: i) routinely diagnosing and reporting SPB via $1-\widehat{\eta}$ for all potential outcome models, including both ML/AI regression and OLS; and ii) adopting UMLR-based causal inference  for unbiased ATE estimation in RWD-based TTE. The main limitations are a modest increase in prediction variance relative to unconstrained ML regression, as the cost of mitigating structural bias, and a currently limited set of implemented UMLR models, both of which will be addressed in future work. 

\section*{Data and Code Availability}
 This research was conducted using the UK Biobank Resource under application number 74376. 
The code used for analyses is available on GitHub at 
\url{https://github.com/effyifeixu/Unbiased-Causal}.

\section*{Acknowledgement}
The authors thank the UK Biobank for their efforts in data collection and curation, which made this research possible.

%-----------------------------------------------------------------------
\bibliographystyle{plainnat}
\bibliography{refs}

\section*{Appendix}
\subsection*{A.1 Proofs of Theoretical Results}
\textit{Theorem~\ref{thm:bias_spb}.} \textit{For each treatment arm $t \in \{0,1\}$, let $\eta_t^{(t)} \in [0,1]$
denote the in-distribution slope parameter and $\eta_t^{(1-t)} \in [0,1]$ denote the
out-of-distribution slope parameter evaluated on the opposite arm. Let $\pi = \mbox{P}(T=1)$. Then the bias of the
estimated ATE via the outcome regression estimators \eqref{eq:mu1_in} and \eqref{eq:mu1_ood} is
\begin{align*}
    \mathrm{Bias}_{\mathrm{SPB}}
    &\;:=\;
    \Bigl(\mathbb{E}_{P(X)}\bigl[\hat{\mu}_1(X)\bigr] - \mathbb{E}_{P(X)}\bigl[\mu_1^*(X)\bigr]\Bigr)
    \;-\;
    \Bigl(\mathbb{E}_{P(X)}\bigl[\hat{\mu}_0(X)\bigr] - \mathbb{E}_{P(X)}\bigl[\mu_0^*(X)\bigr]\Bigr)
    \nonumber\\[6pt]
    &\;=\;
    (1-\pi)\bigl(1-\eta_1^{(0)}\bigr)\,w_1\!\left(\bar{\mu}_1^{(1)}-\bar{\mu}_1^{(0)}\right)
    \;+\;
    \pi\bigl(1-\eta_0^{(1)}\bigr)\,w_0\!\left(\bar{\mu}_0^{(1)}-\bar{\mu}_0^{(0)}\right).
\end{align*}
}

\begin{proof}
Let $\pi = P(T=1)$. The true ATE is
\begin{align*}
    \Delta
    &\;:=\; \mathbb{E}_{P(X)}\bigl[\mu_1^*(X)\bigr] - \mathbb{E}_{P(X)}\bigl[\mu_0^*(X)\bigr],
\end{align*}
where, by the law of total expectation,
\begin{align*}
    \mathbb{E}_{P(X)}\bigl[\mu_1^*(X)\bigr]
    &\;=\; \pi\,\bar{\mu}_1^{(1)} + (1-\pi)\,\bar{\mu}_1^{(0)}, \\[4pt]
    \mathbb{E}_{P(X)}\bigl[\mu_0^*(X)\bigr]
    &\;=\; \pi\,\bar{\mu}_0^{(1)} + (1-\pi)\,\bar{\mu}_0^{(0)},
\end{align*}
where $\bar{\mu}_t^{(s)} := \mathbb{E}[\mu_t^*(X) \mid T=s]$ denotes the true conditional
mean of arm $t$ evaluated on the stratum $T=s$. Under the SPB, since the outcome
models $\hat{\mu}_1$ and $\hat{\mu}_0$ are each trained on a single arm and evaluated out-of-distribution on the opposite arm. We compute the bias of each arm separately.We parametrize the shrinkage targets as
\begin{align*}
    \bar{\mu}_1^{(\alpha_1)} &\;=\; w_1\,\bar{\mu}_1^{(1)} + (1-w_1)\,\bar{\mu}_1^{(0)}, \\[4pt]
    \bar{\mu}_0^{(\alpha_0)} &\;=\; w_0\,\bar{\mu}_0^{(0)} + (1-w_0)\,\bar{\mu}_0^{(1)},
\end{align*}
where $w_1, w_0 \in (0,1)$ are as defined in Section~\ref{sec:biased-ate}.
%where $w_1$ reflects the degree to which $\hat{\mu}_1$ %extrapolates toward the
%training-set mean of the treated arm, and $w_0$ reflects %the degree to which
%$\hat{\mu}_0$ extrapolates toward the training-set mean of %the control arm.

\textbf{Treated arm.}
We first consider $\hat{\mu}_1$, trained on the treated subsample. By the law of total
expectation, conditioning on $T$:
\begin{align*}
    \mathbb{E}_{P(X)}\bigl[\hat{\mu}_1(X)\bigr]
    &= \pi\,\mathbb{E}_{P(X\mid T=1)}\bigl[\hat{\mu}_1(X)\bigr]
     + (1-\pi)\,\mathbb{E}_{P(X\mid T=0)}\bigl[\hat{\mu}_1(X)\bigr].
\end{align*}
Substituting the in-distribution result \eqref{eq:mu1_in} and the OOD result
\eqref{eq:mu1_ood}:
\begin{align*}
    \mathbb{E}_{P(X)}\bigl[\hat{\mu}_1(X)\bigr]
    &= \pi\,\bar{\mu}_1^{(1)}
     + (1-\pi)\Bigl[\eta_1^{(0)}\,\bar{\mu}_1^{(0)}
     + \bigl(1-\eta_1^{(0)}\bigr)\,\bar{\mu}_1^{({\alpha}_1)}\Bigr] \nonumber\\[4pt]
    &= \pi\,\bar{\mu}_1^{(1)}
     + (1-\pi)\,\eta_1^{(0)}\,\bar{\mu}_1^{(0)}
     + (1-\pi)\bigl(1-\eta_1^{(0)}\bigr)\,\bar{\mu}_1^{({\alpha}_1)}.
\end{align*}
Subtracting the true marginal mean
$\mathbb{E}_{P(X)}[\mu_1^*(X)] = \pi\,\bar{\mu}_1^{(1)} + (1-\pi)\,\bar{\mu}_1^{(0)}$:
\begin{align}
    \mathbb{E}_{P(X)}\bigl[\hat{\mu}_1(X)\bigr] - \mathbb{E}_{P(X)}\bigl[\mu_1^*(X)\bigr]
    &= \pi\,\bar{\mu}_1^{(1)}
     + (1-\pi)\,\eta_1^{(0)}\,\bar{\mu}_1^{(0)}
     + (1-\pi)\bigl(1-\eta_1^{(0)}\bigr)\,\bar{\mu}_1^{({\alpha}_1)} \nonumber\\
    &\quad - \pi\,\bar{\mu}_1^{(1)} - (1-\pi)\,\bar{\mu}_1^{(0)} \nonumber\\[4pt]
    %&= (1-\pi)\,\eta_1^{(0)}\,\bar{\mu}_1^{(0)}
    % + (1-\pi)\bigl(1-%\eta_1^{(0)}\bigr)\,\bar{\mu}_1^{({\alpha}_1)}
    % - (1-\pi)\,\bar{\mu}_1^{(0)} \nonumber\\[4pt]
    %&= (1-\pi)\bigl(1-%\eta_1^{(0)}\bigr)\,\bar{\mu}_1^{({\alpha}_1)}
    % - (1-\pi)\bigl(1-%\eta_1^{(0)}\bigr)\,\bar{\mu}_1^{(0)} \nonumber\\[4pt]
    &= (1-\pi)\bigl(1-\eta_1^{(0)}\bigr)
       \Bigl(\bar{\mu}_1^{({\alpha}_1)} - \bar{\mu}_1^{(0)}\Bigr). \label{eq:bias_treated}
\end{align}

\textbf{Control arm.}
We next consider $\hat{\mu}_0$, trained on the control subsample. By the symmetric
argument, the in-distribution result gives
$\mathbb{E}_{P(X\mid T=0)}[\hat{\mu}_0(X)] = \bar{\mu}_0^{(0)}$, and the OOD result gives
$\mathbb{E}_{P(X\mid T=1)}[\hat{\mu}_0(X)] = \eta_0^{(1)}\,\bar{\mu}_0^{(1)} +
(1-\eta_0^{(1)})\,\bar{\mu}_0^{({\alpha}_0)}$. Applying the law of total expectation:
\begin{align}
    \mathbb{E}_{P(X)}\bigl[\hat{\mu}_0(X)\bigr]
    &= \pi\Bigl[\eta_0^{(1)}\,\bar{\mu}_0^{(1)}
     + \bigl(1-\eta_0^{(1)}\bigr)\,\bar{\mu}_0^{({\alpha}_0)}\Bigr]
     + (1-\pi)\,\bar{\mu}_0^{(0)} \nonumber\\[4pt]
    &= \pi\,\eta_0^{(1)}\,\bar{\mu}_0^{(1)}
     + \pi\bigl(1-\eta_0^{(1)}\bigr)\,\bar{\mu}_0^{({\alpha}_0)}
     + (1-\pi)\,\bar{\mu}_0^{(0)}.
\end{align}
Subtracting the true marginal mean
$\mathbb{E}_{P(X)}[\mu_0^*(X)] = \pi\,\bar{\mu}_0^{(1)} + (1-\pi)\,\bar{\mu}_0^{(0)}$:
\begin{align}
    \mathbb{E}_{P(X)}\bigl[\hat{\mu}_0(X)\bigr] - \mathbb{E}_{P(X)}\bigl[\mu_0^*(X)\bigr]
    &= \pi\,\eta_0^{(1)}\,\bar{\mu}_0^{(1)}
     + \pi\bigl(1-\eta_0^{(1)}\bigr)\,\bar{\mu}_0^{({\alpha}_0)}
     + (1-\pi)\,\bar{\mu}_0^{(0)} \nonumber\\
    &\quad - \pi\,\bar{\mu}_0^{(1)} - (1-\pi)\,\bar{\mu}_0^{(0)} \nonumber\\[4pt]
    %&= \pi\,\eta_0^{(1)}\,\bar{\mu}_0^{(1)}
    % + \pi\bigl(1-%\eta_0^{(1)}\bigr)\,\bar{\mu}_0^{({\alpha}_0)}
    % - \pi\,\bar{\mu}_0^{(1)} \nonumber\\[4pt]
    %&= -\,\pi\bigl(1-\eta_0^{(1)}\bigr)\,\bar{\mu}_0^{(1)}
    % + \pi\bigl(1-%\eta_0^{(1)}\bigr)\,\bar{\mu}_0^{({\alpha}_0)} %\nonumber\\[4pt]
    &= -\,\pi\bigl(1-\eta_0^{(1)}\bigr)
       \Bigl(\bar{\mu}_0^{(1)} - \bar{\mu}_0^{({\alpha}_0)}\Bigr). \label{eq:bias_control}
\end{align}

\textbf{Total bias.}
Combining \eqref{eq:bias_treated} and \eqref{eq:bias_control} yields:
\begin{align*}\label{eq:bias_spb}
    \mathrm{Bias}_{\mathrm{SPB}}
    &\;:=\;
    \Bigl(\mathbb{E}_{P(X)}\bigl[\hat{\mu}_1(X)\bigr] - \mathbb{E}_{P(X)}\bigl[\mu_1^*(X)\bigr]\Bigr)
    \;-\;
    \Bigl(\mathbb{E}_{P(X)}\bigl[\hat{\mu}_0(X)\bigr] - \mathbb{E}_{P(X)}\bigl[\mu_0^*(X)\bigr]\Bigr)
    \nonumber\\[6pt]
    &\;=\;
    (1-\pi)\bigl(1-\eta_1^{(0)}\bigr)\!\left(\bar{\mu}_1^{(\alpha_1)}-\bar{\mu}_1^{(0)}\right)
    \;+\;
    \pi\bigl(1-\eta_0^{(1)}\bigr)\!\left(\bar{\mu}_0^{(1)}-\bar{\mu}_0^{(\alpha_0)}\right)
    \nonumber\\[6pt]
    &\;=\;
    (1-\pi)\bigl(1-\eta_1^{(0)}\bigr)\,w_1\!\left(\bar{\mu}_1^{(1)}-\bar{\mu}_1^{(0)}\right)
    \;+\;
    \pi\bigl(1-\eta_0^{(1)}\bigr)\,w_0\!\left(\bar{\mu}_0^{(1)}-\bar{\mu}_0^{(0)}\right).
\end{align*}
\end{proof}
\textit{Remark}.  \textit{ $\mathrm{Bias}_{\mathrm{SPB}} = 0$ in either of the following cases:
\begin{enumerate}[(i)]
    \item $\eta_1^{(0)} = \eta_0^{(1)} = 1$, i.e., both outcome models extrapolate perfectly out of distribution; 
    \item $\bar{\mu}_t^{(1)} = \bar{\mu}_t^{(0)}$ for both $t \in \{0,1\}$, i.e., the conditional outcome surface does not vary across treatment arms (no covariate shift);
    \item $\pi = 0$ and $\bar{\mu}_1^{(1)} = \bar{\mu}_1^{(0)}$, or $\pi = 1$ and $\bar{\mu}_0^{(1)} = \bar{\mu}_0^{(0)}$, and i.e., degenerate treatment allocation with no out-of-distribution evaluation in the other treatment arm.
\end{enumerate}
In practice, conditions (ii) and (iii) are rarely satisfied. We therefore focus on condition (i) as the primary avenue to mitigate $\mathrm{Bias}_{\mathrm{SPB}}$.
}

\textit{Theorem~\ref{thm:main}. }
\textit{For each level of treatment $T=\{0,1\}$, the outcome learner $\hat{\mu}_t^U$ is trained under the constrained objective function
\eqref{eq:obj}. Then the UMLR based ATE estimate in \eqref{eq:ate} is unbiased: $\mathbb{E}(\widehat{\Delta}^U - \Delta)=0$.}
\begin{proof}
The two constraints in \eqref{eq:obj} imply
\begin{equation*}
  \sum_{i \in \mathcal{R}_1}(f^*(X_i) - Y_i) = 0
  \qquad \text{and} \qquad
  \sum_{i \in \mathcal{R}_2}(f^*(X_i) - Y_i) = 0.
\end{equation*}
 
Adding these two equations gives
\begin{align}
  0
  &= \sum_{i \in \mathcal{R}_1}(f^*(X_i) - Y_i)
   + \sum_{i \in \mathcal{R}_2}(f^*(X_i) - Y_i)
   \notag\\
  &= \sum_{i \in \mathcal{R}_1}
       (\mathring{f}(X_i) - c\tilde{f}(X_i) - Y_i)
   + \sum_{i \in \mathcal{R}_2}
       (\mathring{f}(X_i) - c\tilde{f}(X_i) - Y_i)
   \notag\\
  &= \sum_{i=1}^{n}(\mathring{f}(X_i) - Y_i)
   - c\sum_{i=1}^{n}\tilde{f}(X_i),
   \label{eq:key}
\end{align}
where the second equality substitutes the assumed form
$f^*(X_i) = \mathring{f}(X_i) - c\tilde{f}(X_i)$,
with $\mathring{f}(X_i)$ an unbiased prediction of $Y_i$
and $c$ the systematic bias constant.
 
Since $\mathring{f}(X_i)$ is an unbiased prediction of
$Y_i$, taking the expectation of the first term
in~\eqref{eq:key} gives
\begin{equation*}
  \mathbb{E}\!\left[
    \sum_{i=1}^{n}\bigl(\mathring{f}(X_i) - Y_i\bigr)
  \right] = 0.
\end{equation*}
 
Taking the expectation of~\eqref{eq:key} therefore yields
\begin{equation*}
  0 = 0 - c\,\mathbb{E}\!\left[
        \sum_{i=1}^{n}\tilde{f}(X_i)
      \right].
\end{equation*}
 
To ensure $c\,\mathbb{E}\bigl[\sum_{i=1}^{n}
\tilde{f}(X_i)\bigr] = 0$ holds for all cases of
$\sum_{i=1}^{n} f^*(X_i)$, the constant $c$ must be $0$.

 {It follows that $f^*(X_i) = \mathring{f}(X_i)$, so the
predicted potential outcomes satisfy
\begin{equation*}
  \mathbb{E}[f^*(X_i) \mid X_i, T_i] = \mu^*_t(X_i),
  \qquad t \in \{0, 1\},
\end{equation*}
 which is equivalent to $\eta_1^{(1)}=\eta_0^{(0)}=\eta_1^{(0)} = \eta_0^{(1)} = 1$ . }

Therefore, we have the bias from Theorem 1 as 
\begin{align}\label{eq:bias_spb}
    \mathrm{Bias}_{\mathrm{UMLR}}
    &\;:=\;
    \Bigl(\mathbb{E}_{P(X)}\bigl[\hat{\mu}_1(X)\bigr] - \mathbb{E}_{P(X)}\bigl[\mu_1^*(X)\bigr]\Bigr)
    \;-\;
    \Bigl(\mathbb{E}_{P(X)}\bigl[\hat{\mu}_0(X)\bigr] - \mathbb{E}_{P(X)}\bigl[\mu_0^*(X)\bigr]\Bigr)
    \nonumber\\[6pt]
    &\;=\;
    (1-\pi)\bigl(1-\eta_1^{(0)}\bigr)\,w_1\!\left(\bar{\mu}_1^{(1)}-\bar{\mu}_1^{(0)}\right)
    \;+\;
    \pi\bigl(1-\eta_0^{(1)}\bigr)\,w_0\!\left(\bar{\mu}_0^{(1)}-\bar{\mu}_0^{(0)}\right)
    \nonumber\\[6pt]
     &\;=\;
     0
\end{align}
\end{proof}
\textit{Remark}.  \textit{ More generally, since $\bar{\mu}_1^{(1)} - \bar{\mu}_1^{(0)}$ and
$\bar{\mu}_0^{(1)} - \bar{\mu}_0^{(0)}$ share the same sign under typical covariate
shift — both reflecting the same underlying distributional difference between
$P(X \mid T=1)$ and $P(X \mid T=0)$ — the magnitude of $\mathrm{Bias}_{\mathrm{SPB}}$
decreases monotonically as $\eta_1^{(0)}$ and $\eta_0^{(1)}$ approach $1$. Hence,
even when perfect extrapolation ($\eta_1^{(0)} = \eta_0^{(1)} = 1$) cannot be achieved,
the bias becomes smaller as $\eta_1^{(0)}$ and $\eta_0^{(1)}$ are closer to $1$. }

\subsection*{A.2 The Impact of  Systematic Prediction Bias on Doubly Robust Estimator}

Doubly robust estimators are powerful tools for causal inference because they can provide unbiased estimates when either the propensity score model or the outcome regression model is correctly specified. However, in settings with high-dimensional confounding, the outcome regression component is often modeled using AI/ML regression, which  incurs systematic prediction bias \citep{lee2025systematic}. In this case, the doubly robust estimator loses one layer of protection and must rely entirely on correct specification of the propensity score model. This reliance can be problematic because propensity score models are known to face several challenges in practice, including instability, limited overlap, and difficulty in accurately modeling treatment assignment under high-dimensional confounding \citep{schneeweiss2009high,li2019addressing,zhou2020propensity,rassen2022high}. Together, these issues can lead to biased ATE estimation. Specifically, the following theorem formally establishes the bias of the ATE in AIPW model.
\begin{theorem}
Let the augmented inverse probability weighted (AIPW) estimator is:
\begin{equation*}
    \hat{\Delta}_{\mathrm{AIPW}}
    = \frac{1}{n}\sum_{i=1}^n \left[
        \hat{\mu}_1(X_i) - \hat{\mu}_0(X_i)
        + \frac{T_i\bigl(Y_i - \hat{\mu}_1(X_i)\bigr)}{\hat{e}(X_i)}
        - \frac{(1-T_i)\bigl(Y_i - \hat{\mu}_0(X_i)\bigr)}{1-\hat{e}(X_i)}
    \right]. 
\end{equation*}
Then, under the settings that potential outcomes are estimated by ML/AI regression with SPB, we have  
\begin{align*}\label{eq:bias_spb}
    \mathrm{Bias}_{\mathrm{AIPW}}
    &\;=\;
    (1-\pi)\bigl(1-\eta_1^{(0)}\bigr)\,w_1\!\left(\bar{\mu}_1^{(1)}-\bar{\mu}_1^{(0)}\right)
    \;+\;
    \pi\bigl(1-\eta_0^{(1)}\bigr)\,w_0\!\left(\bar{\mu}_0^{(1)}-\bar{\mu}_0^{(0)}\right).
\end{align*}
\end{theorem}

\begin{proof}
    In population, assuming $\hat{e}(x) = e^*(x)$ (propensity correctly specified),
the estimator converges to:
\begin{equation}
    \hat{\Delta}_{\mathrm{AIPW}} \xrightarrow{p}
    \underbrace{\mathbb{E}[\hat{\mu}_1(X) - \hat{\mu}_0(X)]}_{\Delta + Bias_{\mathrm{OR}}}
    + \underbrace{\mathbb{E}\!\left[\frac{T(Y-\hat{\mu}_1(X))}{e^*(X)}\right]}_{A_1}
    - \underbrace{\mathbb{E}\!\left[\frac{(1-T)(Y-\hat{\mu}_0(X))}{1-e^*(X)}\right]}_{A_0}.
\end{equation}

\paragraph{Treated term $A_1$.}
Conditioning on $X$ and using $\mathbb{E}[T \mid X] = e^*(X)$:
\begin{align*}
    A_1 &= \mathbb{E}_X\!\left[\mu_1^*(X) - \hat{\mu}_1(X) \mid T=1\right] \nonumber \\
        &= (1-\eta_1^{(1)})\,\mathbb{E}\!\left[\mu_1^*(X) - \bar{\mu}_1^{(1)} \mid T=1\right]
         = 0.
\end{align*}
The residuals for $T=1$ units are mean-zero over $P(X \mid T=1)$ because
$\hat{\mu}_1$ is marginally calibrated on its training distribution.

\paragraph{Control term $A_0$.}
By the symmetric argument over $P(X \mid T=0)$:
\begin{align*}
    A_0 &= \mathbb{E}_X\!\left[\mu_0^*(X) - \hat{\mu}_0(X) \mid T=0\right] \nonumber \\
        &= (1-\eta_0^{(0)})\,\mathbb{E}\!\left[\mu_0^*(X) - \bar{\mu}_0^{(0)} \mid T=0\right]
         = 0.
\end{align*}

Since $A_1 - A_0 = 0$, the AIPW estimator satisfies:
\begin{equation*}
{
    \hat{\Delta}_{\mathrm{AIPW}} \xrightarrow{p}
    \Delta
    + (1-\pi)\bigl(1-\eta_1^{(0)}\bigr)\,w_1\!\left(\bar{\mu}_1^{(1)}-\bar{\mu}_1^{(0)}\right)
    \;+\;
    \pi\bigl(1-\eta_0^{(1)}\bigr)\,w_0\!\left(\bar{\mu}_0^{(1)}-\bar{\mu}_0^{(0)}\right).
}
\end{equation*}

That is, 
\begin{equation*}
    \hat{\Delta}_{\mathrm{AIPW}}
    = \frac{1}{n}\sum_{i=1}^n \left[
        \hat{\mu}_1(X_i) - \hat{\mu}_0(X_i)
        + \frac{T_i\bigl(Y_i - \hat{\mu}_1(X_i)\bigr)}{\hat{e}(X_i)}
        - \frac{(1-T_i)\bigl(Y_i - \hat{\mu}_0(X_i)\bigr)}{1-\hat{e}(X_i)}
    \right]. 
\end{equation*}

\end{proof}

\textit{Remark}.  \textit{ The AIPW estimator is exactly as biased as the Outcome Regression estimator under SPB 
even when the propensity score is correctly specified.}

\subsection*{A.3 Simulation Study: Finite-Sample Bias in Doubly Robust Estimation}
 {We conducted additional simulation studies to evaluate the finite-sample performance of doubly robust estimators when conventional MLR is employed as the outcome regression models for potential outcome prediction with SPB. Specifically, we assess the influence of systematic prediction bias on ATE estimation bias. In addition, we examined the performance of doubly robust model implemented by  the double machine learning (DML) estimator \citep{chernozhukov2018double} }

 {Using the data-generating procedure of Section~\ref{sec:sim}, we simulated datasets with $p = 200$ covariates, across sample sizes $n \in \{100, 500, 2{,}000, 5{,}000\}$. }

 \paragraph{The Influence of Systematic Prediction Bias on ATE Estimation in Doubly Robust Estimators.}

To attribute bias in doubly robust ATE estimates solely to systematic prediction bias in the outcome model, we use the true propensity score, rather than an ML-estimated counterpart, as the oracle treatment assignment probability. Since doubly robust ATE bias vanishes asymptotically, we evaluate finite-sample bias across $n \in \{100, 500, 2{,}000, 5{,}000\}$, where $n = 5{,}000$ represents a practically large sample. We further consider two noise levels, $\sigma = 5$ and $\sigma = 1$.   }

 {As shown in Figure~\ref{fig:fig4}(a), doubly robust ATE bias decreases monotonically with sample size, yet convergence to zero is slow: non-negligible bias remains even at $n = 5{,}000$. Counterintuitively, higher noise levels reduce bias. A possible explanation is that elevated noise attenuates the contribution of the outcome regression component to the doubly robust estimator, thereby dampening the influence of systematic prediction bias on ATE estimation. This pattern underscores that conventional MLR-based doubly robust estimators are most sensitive to SPB in high signal-to-noise settings, the very conditions under which precise treatment effect estimates are most consequential for clinical decision-making.  }

\paragraph{Bias of ATE Estimation by MLR- vs. UMLR-Based Doubly Robust Estimators.}
We further compare the bias in ATE estimation between MLR- and UMLR-based doubly robust estimators. Specifically, at $n = 500$, we evaluate the ATE bias for the RCT estimator, the UMLR-based AIPW estimator, the MLR-based AIPW estimator, and the MLR-based DML estimator. As shown in Figure~\ref{fig:fig4}(b), non-negligible ATE bias persists for the MLR-based AIPW and DML estimators, whereas the UMLR-based AIPW estimator achieves nearly unbiased ATE estimation, comparable to the gold-standard RCT estimate. These results suggest that systematic prediction bias in MLR can lead to biased ATE estimation for both AIPW and DML estimators, although DML may partially mitigate this bias. In contrast, UMLR can substantially reduce ATE estimation bias by reducing systematic prediction bias, which has a direct impact on ATE estimation. Because UMLR models are compatible with existing MLR-based doubly robust estimators, they provide an alternative avenue for causal inference using RWD.

\begin{figure}[ht]
  \centering
  \includegraphics[width=1\linewidth]{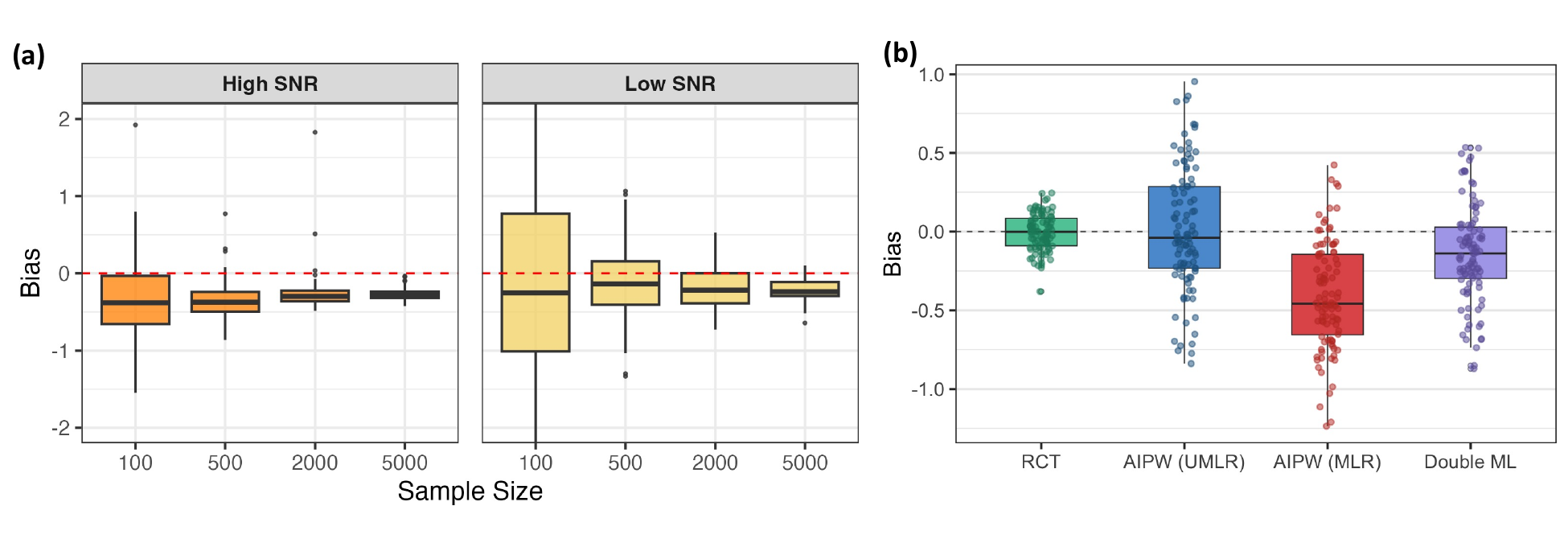}
  \caption{Influence of systematic prediction bias on doubly robust ATE estimation. \textbf{(a)} With oracle propensity scores, ATE bias decreases monotonically with sample size but converges slowly; non-negligible bias remains at $n = 5{,}000$. Higher noise levels attenuate SPB-induced bias in the outcome regression. \textbf{(b)} ATE bias persists in MLR-based AIPW and DML estimators, while UMLR-based AIPW estimator achieves nearly unbiased estimation. }  
  \label{fig:fig4}
\end{figure}
\paragraph{Unbiased ATE by UMLR-based Causal Inference Models.}
Under the ideal UMLR setting, where $(\eta_1^{(1)}=\eta_0^{(0)}=\eta_1^{(0)}=\eta_0^{(1)}=1)$, the UMLR-based doubly robust estimator is unbiased.
In practice, even when UMLR yields $(\eta_t^{(s)}<1)$, the resulting SPB is typically much smaller than that produced by conventional ML/AI regression. Therefore, the bias in ATE estimation can be substantially reduced, as discussed in the remark following Theorem~\ref{thm:main}.

% $\blacksquare$
% \paragraph{Bias--variance trade-off.}
% Imposing the constraints in \eqref{eq:obj} reduces bias at the cost of potentially increased estimator
% variance, since the feasible set is smaller than that of the unconstrained problem. In finite
% samples this may raise the MSE of $f^*$ relative to the unconstrained solution. However, because
% the ATE estimator is a population-level average of $f^*(X)$, the variance inflation is
% $O(n^{-1/2})$ and vanishes as $n \to \infty$, whereas the bias of the unconstrained estimator
% is $O(1)$ and does not diminish. The constrained solution therefore dominates for the inferential
% task of ATE estimation whenever $n$ is moderately large.

\subsubsection*{A.4 UKB RWD Analysis Supplementary Results and Evaluation Metrics}

 {Appendix~Table~\ref{tab:descriptive} summarizes baseline covariates   for participants with chronic pain in the UK Biobank, stratified by opioid use status. The analytic cohort included 19,736 participants, of whom 523 (2.6\%) reported current opioid use and 19,213 (97.4\%) were non-users, indicating substantial imbalance in treatment group sizes. Compared with non-users, opioid users were older, had higher body mass index, greater socioeconomic deprivation, lower educational attainment, and were more likely to be current or former smokers. They also exhibited less favorable inflammatory and metabolic biomarker profiles, including higher levels of C-reactive protein, glucose, HbA1c, triglycerides, and cystatin C, and lower levels of HDL cholesterol, IGF-1, SHBG, and vitamin D. The substantial differences in baseline characteristics highlight the needs of for careful confounding adjustment when estimating the causal effect of opioid use on blood pressure.}

 {After propensity score matching, 518 opioid users were successfully matched to 518 non-users. Baseline covariates were well balanced between the two groups, with insignificant $p$-values and minimal standardized mean differences (generally below 0.1) for all covariates (Appendix~Table~\ref{tab:psm_descriptive}).}

 {Appendix~Table~\ref{tab:eval} evaluates the performance for MLR and UMLR under the S- and T-learners in the UK Biobank chronic pain cohort. Consistent with the simulation results, MLR exhibited noticeable systematic prediction bias, with bias slopes of approximately 0.85 for both learners and strong correlations between observed and prediction residuals ($\mathrm{Cor}(y,\hat{\varepsilon}) \approx -0.79$). In contrast, UMLR substantially reduced prediction bias, with slopes close to zero and near-zero residual correlations, indicating effective removal of systematic prediction bias. UMLR produced larger RMSE and MAE than MLR for both learners. Corresponding diagnostics are not reported for the X-learner because its second-stage models target pseudo-outcomes rather than the observed outcome directly.}

\begin{table}[ht]
\centering
\caption{Descriptive statistics of covariates stratified by opioid use status among participants with chronic pain in the UK Biobank.}
\label{tab:descriptive}
\setlength{\tabcolsep}{14.5pt}
\small
\begin{tabular}{lrrr}
\toprule
\textbf{Characteristic} 
    & \textbf{Non-users} 
    & \textbf{Opioid users} 
    & \textbf{$p$-value} \\
    & \textbf{($n = 19{,}213$)} 
    & \textbf{($n = 523$)} 
    & \\
\midrule
\multicolumn{4}{l}{\textit{Demographic and lifestyle characteristics}} \\
Sex = male, \%   
    & 8384 (43.6)     & 218 (41.7)      & 0.398    \\
\addlinespace
Age, mean (SD) 
    & 53.57 (7.28)    & 56.53 (7.53)    & $<$0.001 \\
    
\addlinespace
Ethnicity, \%                        &              &            & 0.797 \\
\quad 1 (White)                      & 18482 (96.2) & 509 (97.3) &       \\
\quad 2 (Mixed)                      & 124 (0.6)    & 3 (0.6)    &       \\
\quad 3 (Asian or Asian British)     & 261 (1.4)    & 6 (1.1)    &       \\
\quad 4 (Black or Black British)     & 150 (0.8)    & 2 (0.4)    &       \\
\quad 5 (Chinese)                    & 53 (0.3)     & 1 (0.2)    &       \\
\quad 6 (Other ethnic group)         & 143 (0.7)    & 2 (0.4)    &       \\

Education, \%    &               &               & $<$0.001 \\
\quad 1  (College or University degree)        & 9244 (48.1)   & 181 (34.6)    &          \\
\quad 2  (A levels/AS levels or equivalent)        & 2798 (14.6)   & 72 (13.8)     &          \\
\quad 3  (O levels/GCSEs or equivalent)        & 4199 (21.9)   & 157 (30.0)    &          \\
\quad 4  (CSEs or equivalent)        & 977 (5.1)     & 35 (6.7)      &          \\
\quad 5  (NVQ or HND or HNC or equivalent)        & 1068 (5.6)    & 43 (8.2)      &          \\
\quad 6  (Other professional qualifications eg: nursing, teaching)        & 927 (4.8)     & 35 (6.7)      &          \\

\addlinespace
Smoking status, \%    &              &             & $<$0.001 \\
\quad 0 (Never)       & 11336 (59.0) & 259 (49.5)  &          \\
\quad 1 (Previous)    & 6499 (33.8)  & 210 (40.2)  &          \\
\quad 2 (Current)     & 1378 (7.2)   & 54 (10.3)   &          \\

\addlinespace
Alcohol drinking status, \%  &             &            & $<$0.001 \\
\quad 0 (Never)              & 532 (2.8)   & 17 (3.3)   &          \\
\quad 1 (Previous)           & 493 (2.6)   & 31 (5.9)   &          \\
\quad 2 (Current)            & 18188 (94.7)& 475 (90.8) &          \\

\midrule

\multicolumn{4}{l}{\textit{Anthropometric and socioeconomic measures}} \\
BMI, mean (SD) 
    & 26.60 (4.34)    & 29.65 (5.71)    & $<$0.001 \\
Townsend Deprivation Index, mean (SD) 
    & $-$1.83 (2.76)  & $-$1.12 (2.98)  & $<$0.001 \\

\midrule
\multicolumn{4}{l}{\textit{Blood biomarkers}} \\
Albumin, mean (SD)                      
    & 45.44 (2.44)    & 44.94 (2.44)    & $<$0.001 \\
Alkaline phosphatase, mean (SD)         
    & 78.99 (22.75)   & 83.21 (24.30)   & $<$0.001 \\
Alanine aminotransferase, mean (SD)     
    & 22.56 (13.16)   & 24.76 (13.50)   & $<$0.001 \\
Apolipoprotein A, mean (SD)             
    & 1.55 (0.25)     & 1.49 (0.25)     & $<$0.001 \\
Apolipoprotein B, mean (SD)             
    & 1.03 (0.22)     & 1.04 (0.23)     & 0.494    \\
Aspartate aminotransferase, mean (SD)   
    & 25.28 (8.65)    & 26.10 (9.01)    & 0.034    \\
Direct bilirubin, mean (SD)             
    & 1.73 (0.74)     & 1.58 (0.59)     & $<$0.001 \\
Urea, mean (SD)                         
    & 5.22 (1.19)     & 5.47 (1.32)     & $<$0.001 \\
Calcium, mean (SD)                      
    & 2.38 (0.09)     & 2.37 (0.09)     & 0.100    \\
Cholesterol, mean (SD)                  
    & 5.71 (1.03)     & 5.59 (1.13)     & 0.008    \\
Creatinine, mean (SD)                   
    & 71.39 (15.09)   & 70.81 (14.15)   & 0.384    \\
C-reactive protein, mean (SD)           
    & 2.04 (3.47)     & 3.39 (4.62)     & $<$0.001 \\
Cystatin C, mean (SD)                   
    & 0.86 (0.13)     & 0.91 (0.14)     & $<$0.001 \\
Gamma glutamyltransferase, mean (SD)    
    & 32.47 (30.64)   & 44.54 (54.35)   & $<$0.001 \\
Glucose, mean (SD)                      
    & 4.94 (0.78)     & 5.23 (1.80)     & $<$0.001 \\
Glycated haemoglobin (HbA1c), mean (SD) 
    & 34.66 (4.55)    & 36.30 (7.24)    & $<$0.001 \\
HDL cholesterol, mean (SD)              
    & 1.48 (0.38)     & 1.37 (0.36)     & $<$0.001 \\
IGF-1, mean (SD)                        
    & 22.11 (5.31)    & 20.48 (5.47)    & $<$0.001 \\
LDL direct, mean (SD)                   
    & 3.58 (0.79)     & 3.53 (0.85)     & 0.175    \\
Phosphate, mean (SD)                    
    & 1.16 (0.15)     & 1.14 (0.15)     & 0.004    \\
SHBG, mean (SD)                         
    & 53.37 (28.23)   & 47.28 (25.80)   & $<$0.001 \\
Total bilirubin, mean (SD)              
    & 9.23 (4.33)     & 8.08 (3.34)     & $<$0.001 \\
Testosterone, mean (SD)                 
    & 5.94 (6.00)     & 5.34 (5.65)     & 0.022    \\
Total protein, mean (SD)                
    & 72.27 (3.75)    & 71.91 (3.80)    & 0.031    \\
Triglycerides, mean (SD)                
    & 1.62 (0.94)     & 1.95 (1.20)     & $<$0.001 \\
Urate, mean (SD)                        
    & 298.19 (75.97)  & 310.25 (83.09)  & $<$0.001 \\
Vitamin D, mean (SD)                    
    & 50.16 (20.20)   & 47.44 (21.40)   & 0.002    \\
\bottomrule
\end{tabular}
\begin{tablenotes}
\small
\item BMI = body mass index; SHBG = sex hormone-binding globulin; IGF-1 = insulin-like growth factor 1; HbA1c = glycated haemoglobin; HDL = high-density lipoprotein; LDL = low-density lipoprotein; SBP = systolic blood pressure; SD = standard deviation. $p$-values are from two-sample $t$-tests for continuous variables and chi-squared tests for categorical variables.
\end{tablenotes}
\end{table}

%\tblue{ZY: the descriptive statistics table for PSM particiapnts was added Table~\ref{tab:psm_descriptive}. }

\begin{table}[ht]
\centering
\caption{Descriptive statistics of covariates stratified by
opioid use status among propensity score matched participants with
chronic pain in the UK Biobank.}
\label{tab:psm_descriptive}
\small
\setlength{\tabcolsep}{14.5pt}
\begin{tabular}{lrrr}
\toprule
\textbf{Characteristic}
    & \textbf{Non-users}
    & \textbf{Opioid users}
    & \textbf{$p$-value} \\
    & \textbf{($n = 518$)}
    & \textbf{($n = 518$)}
    & \\
\midrule

\multicolumn{4}{l}{\textit{Demographic and lifestyle characteristics}} \\
Sex = male, \%
    & 205 (39.6)   & 217 (41.9)   & 0.487 \\

\addlinespace
Age, mean (SD)
    & 56.27 (7.42) & 56.47 (7.53) & 0.660 \\

\addlinespace
Ethnicity, \%                              &              &              & 0.388 \\
\quad 1 (White)                            & 510 (98.5)   & 504 (97.3)   &       \\
\quad 2 (Mixed)                            & 0 (0.0)      & 3 (0.6)      &       \\
\quad 3 (Asian or Asian British)           & 3 (0.6)      & 6 (1.2)      &       \\
\quad 4 (Black or Black British)           & 2 (0.4)      & 2 (0.4)      &       \\
\quad 5 (Chinese)                          & 0 (0.0)      & 1 (0.2)      &       \\
\quad 6 (Other ethnic group)               & 3 (0.6)      & 2 (0.4)      &       \\

\addlinespace
Education, \%                              &              &              & 0.979 \\
\quad 1 (College or University degree)     & 176 (34.0)   & 180 (34.7)   &       \\
\quad 2 (A levels/AS levels or equivalent) & 70 (13.5)    & 70 (13.5)    &       \\
\quad 3 (O levels/GCSEs or equivalent)    & 158 (30.5)   & 156 (30.1)   &       \\
\quad 4 (CSEs or equivalent)              & 37 (7.1)     & 35 (6.8)     &       \\
\quad 5 (NVQ or HND or HNC or equivalent) & 37 (7.1)     & 42 (8.1)     &       \\
\quad 6 (Other professional qualifications eg: nursing, teaching)
                                           & 40 (7.7)     & 35 (6.8)     &       \\

\addlinespace
Smoking status, \%                         &              &              & 0.838 \\
\quad 0 (Never)                            & 265 (51.2)   & 256 (49.4)   &       \\
\quad 1 (Previous)                         & 199 (38.4)   & 208 (40.2)   &       \\
\quad 2 (Current)                          & 54 (10.4)    & 54 (10.4)    &       \\

\addlinespace
Alcohol drinking status, \%                &              &              & 0.914 \\
\quad 0 (Never)                            & 19 (3.7)     & 17 (3.3)     &       \\
\quad 1 (Previous)                         & 28 (5.4)     & 30 (5.8)     &       \\
\quad 2 (Current)                          & 471 (90.9)   & 471 (90.9)   &       \\

\midrule
\multicolumn{4}{l}{\textit{Anthropometric and socioeconomic measures}} \\
BMI, mean (SD)
    & 29.59 (6.04)    & 29.51 (5.54)    & 0.820 \\
Townsend Deprivation Index, mean (SD)
    & $-$1.11 (3.07)  & $-$1.13 (2.99)  & 0.938 \\

\midrule
\multicolumn{4}{l}{\textit{Blood biomarkers}} \\
Albumin, mean (SD)
    & 44.94 (2.46)    & 44.94 (2.43)    & 0.999 \\
Alkaline phosphatase, mean (SD)
    & 82.47 (22.87)   & 83.25 (24.36)   & 0.598 \\
Alanine aminotransferase, mean (SD)
    & 26.05 (18.20)   & 24.80 (13.54)   & 0.211 \\
Apolipoprotein A, mean (SD)
    & 1.48 (0.24)     & 1.49 (0.25)     & 0.642 \\
Apolipoprotein B, mean (SD)
    & 1.04 (0.23)     & 1.04 (0.23)     & 0.927 \\
Aspartate aminotransferase, mean (SD)
    & 27.62 (20.85)   & 26.12 (9.03)    & 0.132 \\
Direct bilirubin, mean (SD)
    & 1.56 (0.56)     & 1.58 (0.59)     & 0.553 \\
Urea, mean (SD)
    & 5.46 (1.38)     & 5.46 (1.32)     & 0.976 \\
Calcium, mean (SD)
    & 2.37 (0.09)     & 2.37 (0.09)     & 0.515 \\
Cholesterol, mean (SD)
    & 5.58 (1.03)     & 5.60 (1.13)     & 0.800 \\
Creatinine, mean (SD)
    & 70.09 (14.20)   & 70.90 (14.12)   & 0.353 \\
C-reactive protein, mean (SD)
    & 3.20 (5.36)     & 3.37 (4.63)     & 0.579 \\
Cystatin C, mean (SD)
    & 0.90 (0.14)     & 0.91 (0.14)     & 0.146 \\
Gamma glutamyltransferase, mean (SD)
    & 44.70 (59.50)   & 44.35 (54.48)   & 0.923 \\
Glucose, mean (SD)
    & 5.16 (1.19)     & 5.21 (1.77)     & 0.591 \\
Glycated haemoglobin (HbA1c), mean (SD)
    & 36.15 (6.29)    & 36.19 (7.08)    & 0.916 \\
HDL cholesterol, mean (SD)
    & 1.36 (0.35)     & 1.37 (0.36)     & 0.601 \\
IGF-1, mean (SD)
    & 20.49 (5.20)    & 20.52 (5.48)    & 0.920 \\
LDL direct, mean (SD)
    & 3.53 (0.79)     & 3.53 (0.85)     & 0.981 \\
Phosphate, mean (SD)
    & 1.14 (0.15)     & 1.14 (0.15)     & 0.836 \\
SHBG, mean (SD)
    & 45.99 (24.52)   & 47.48 (25.78)   & 0.341 \\
Total bilirubin, mean (SD)
    & 8.14 (3.17)     & 8.09 (3.35)     & 0.815 \\
Testosterone, mean (SD)
    & 5.06 (5.31)     & 5.35 (5.65)     & 0.382 \\
Total protein, mean (SD)
    & 71.72 (3.57)    & 71.88 (3.75)    & 0.505 \\
Triglycerides, mean (SD)
    & 1.92 (1.04)     & 1.93 (1.14)     & 0.812 \\
Urate, mean (SD)
    & 310.32 (81.60)  & 309.65 (83.17)  & 0.896 \\
Vitamin D, mean (SD)
    & 47.56 (19.82)   & 47.63 (21.40)   & 0.957 \\
\bottomrule
\end{tabular}
\begin{tablenotes}
\small
\item BMI = body mass index; SHBG = sex hormone-binding globulin;
IGF-1 = insulin-like growth factor 1; HbA1c = glycated haemoglobin;
HDL = high-density lipoprotein; LDL = low-density lipoprotein;
SBP = systolic blood pressure; DBP = diastolic blood pressure;
SD = standard deviation. $p$-values are from two-sample $t$-tests for continuous variables and chi-squared tests for categorical variables.
%Propensity score matching was performed using 1:1 nearest-neighbour matching without replacement.
\end{tablenotes}
\end{table}

\clearpage
\vspace*{-1cm}
\begin{table}[ht]
\centering
\caption{Evaluation of ATE estimators for the effect of opioid use on systolic blood
pressure in the UK Biobank chronic pain cohort (testing set).}
\label{tab:eval}
\setlength{\tabcolsep}{9pt}
\begin{tabular}{l cccc cccc}
\toprule
& \multicolumn{4}{c}{S-learner} & \multicolumn{4}{c}{T-learner} \\
\cmidrule(lr){2-5} \cmidrule(lr){6-9}
Method & Bias (Slope) & RMSE & MAE & $\mathrm{Cor}({\mathbf{y}}, \hat{\boldsymbol{\varepsilon}})$ & Bias (Slope) & RMSE & MAE & $\mathrm{Cor}({\mathbf{y}}, \hat{\boldsymbol{\varepsilon}})$ \\
\midrule
MLR  & 0.852 & 20.92 & 16.19 & -0.79 & 0.853 & 21.00 & 16.23 & -0.79 \\
UMLR & 0.097 & 55.88 & 45.10 & -0.03 & 0.096 & 55.22 & 44.62 & -0.03 \\
\bottomrule
\end{tabular}
\end{table}

\end{document}